\documentclass[sim,STIX1COL,doublespace]{WileyNJD-v2}
\RequirePackage{epstopdf,graphicx,amsopn,amsmath,color}

\newcommand{\mini}{\textrm{minimize}}
\newcommand{\sign}{\textrm{sign}}
\newcommand{\adapt}[1]{\textcolor{black}{#1}}

\articletype{Article Type}%

\received{26 April 2016}
\revised{6 June 2016}
\accepted{6 June 2016}

\begin{document}

\title{Estimating Individualized Optimal Combination Therapies through Outcome Weighted Deep Learning Algorithms}

\author[1]{MUXUAN LIANG*}

\author[2]{TING YE}

\author[3]{HAODA FU*}

\authormark{Liang, Ye and Fu}

\address[1,2]{\orgdiv{Department of Statistics}, \orgname{University of Wisconsin-Madison}, \orgaddress{\state{Wisconsin}, \country{U.S.A.}}}

\address[3]{\orgname{Eli Lilly and Company}, \orgaddress{\state{Indiana}, \country{U.S.A.}}}

\corres{*Correspondence to: Haoda Fu, Eli Lilly and Company, IN,
	U.S.A.
	\email{fu\_haoda@lilly.com}\\	
	*Correspondence to: Muxuan Liang, University of Wisconsin-Madison, WI,
	U.S.A.
	\email{mliang9@wisc.edu}}


\abstract[Summary]{ With the advancement in drug development, multiple treatments are available for a single disease. Patients can often benefit from taking multiple treatments simultaneously. For example, patients in Clinical Practice Research Datalink (CPRD) with chronic diseases such as type 2 diabetes can receive multiple treatments simultaneously. Therefore, it is important to estimate what combination therapy from which patients can benefit the most. However, to recommend the best treatment combination is not a single-label but a multi-label classification problem.  In this paper, we propose a novel outcome weighted deep learning algorithm to estimate individualized optimal combination therapy. The fisher consistency of the proposed loss function under certain conditions is also provided. In addition, we extend our method to a family of loss functions, which allows adaptive changes based on treatment interactions. We demonstrate the performance of our methods through simulations and real data analysis.}


\maketitle

\section{Introduction}\label{sec1}

In recent years, precision medicine has become an important topic in both industry and academia. It aims to find a mechanism of treatment assignment such that the patient can benefit the most. A lot of researchers have conducted vast investigations through different approaches. For example, the framework of outcome weighted learning is proposed to identify the individualized treatment rule (ITR) in \cite{Zhao2012}. Alternatively, Zhang et al. \cite{Zhang2012a} \cite{Zhang2012b} propose another general framework of estimating optimal treatment regimes with robustness from the perspective of classification. Fu, Zhou and Faries \cite{Fu2016} develop a comprehensive binary search approach, which is easy to interpret and apply in clinical study. Chen \cite{Chen2017} further extends the outcome weighted learning to other models and loss functions. However, these methods only focus on binary treatment recommendation. To handle multiple treatments, Zhang et al. \cite{Zhang2017} develop a multi-category outcome weighted learning approach using angle based classifiers.  
Their model essentially assumes that each patient takes one out of multiple treatments. However, in reality, patients can benefit from taking multiple treatments simultaneously. Therefore, it is of great interest to develop methods for precision medicine in the context of combination therapies.

The adoption of combination therapies in medication is inevitable, especially in chronic diseases. For example, patients with type 2 diabetes often receive multiple medications, because single treatment may be insufficient to effectively control the blood glucose level. Hence, medications exerting their effects through different mechanisms are needed. For instance, DPP4 increases incretin level, which inhibits glucagon release. Sulfonylurea increases insulin release from $\beta$-cell in pancreas. DDP4 and Sulfonylurea function through different biological pathways, while serving the same purpose of reducing blood glucose. Thus, we can expect better control in blood glucose when taking both medications. In addition, diabetes may also cause other complications which require additional medications. Due to these two reasons, patients with type 2 diabetes are very likely to receive a combination of multiple treatments, and so is the case in other chronic diseases such as cancer, and chronic heart failure (CHF).

However, existing algorithms may not be adequate to precision medicine with combination therapies. Algorithms for multi-class classification are not scalable with the increasing number of treatments, that is, if we have $K$ treatments, there are $2^K$ possible combinations, which implies at least $2^K-1$ classifiers. This is the consequence of treating each combination individually \adapt{without considering the underlying structures among those treatments}. For example, DDP4 and Sulfonylurea can reduce A1c by $0.5\%$ and $1.0 \%$, respectively. \adapt{If the treatment effects of DDP4 and Sulfonylurea are additive, combining these two medications together can reduce A1c by $1.5\%$ . If small interaction exists, the reduction is expected to lie between $1.0\%$ and $2.0\%$.}  {\adapt Our proposed method is able to leverage this information to improve the learning efficiency, which will be elaborated in Section \ref{sec5} and \ref{sec6}. }

Different from the traditional supervised learning, the estimand is the optimal treatment assignment which is not directly observed from the patients. For example, in a classification problem, the correct label $Y$ is observed for each observed covariates $X$. Similarly, in a regression problem, an outcome $Y$ is also observed for each observed covariates $X$. In our treatment assignment problem, information about optimal treatment assignment rule is only available indirectly through the outcome due to the fact that only one potential outcome can be observed. In this paper, we extend deep learning algorithms, and propose an outcome weighted loss function to estimate optimal treatment assignment rule for combination therapies. 

Our paper is organized in the following manner. Section \ref{sec2} provides a brief review to outcome weighted learning and multi-label classification problem. It explains the necessity of using multi-label classification rather than multi-class classification, and also {\adapt addresses some requirements that a good method needs to achieve}. Section \ref{sec3} incorporates multi-label classification with outcome weighted learning, and illustrates basic properties of our proposed method. Section \ref{sec4} further provides theoretical justifications on the fisher consistency under general cases and a special case with additive treatment effect and small interactions. Section \ref{sec5} and \ref{sec6} show the advantages of our method through simulations and real data analysis. In section \ref{sec7}, our method is further extended to a family of loss function which can be adaptive to the treatment interactions. We also discuss other possible future extensions in section \ref{sec7}.

\section{Review}\label{sec2}
\label{review:multi-label}
In this section, we will review some literature related to precision medicine and multi-label classification. In the meantime, the keys to solve this problem are identified.

Precision medicine has been a hot topic for years. On the one hand, outcome weighted learning proposed in \cite{Zhao2012} is one of the most popular method. Outcome weighted learning finds the optimal decision rule that maximizes the conditional expectation of the outcome given the decision rule. Let $X$ be the covariates, function $D(\cdot)$ be the decision rule, which is a mapping from space of covariates $\mathcal{X}$ to the space of treatment $\mathcal{A}$. The conditional expectation of the outcome $R\in \mathbb{R}$ given treatment assignment $D(X)$ can be written as:
\begin{equation}\label{eq:owl}
E^{D}(R)=E\left[\frac{R}{\pi_A}I\{A=D(X)\}\right],
\end{equation}
where $\pi_a=\Pr(A=a|X)$, where $A$ is the random variable of treatment assignment and $a\in \mathcal{A}$ is a treatment. In this framework, searching for the best treatment assignment $D$ to maximize $E^{D}(R)$ is converted into a classification problem. {\adapt Similar to the multi-class classification problem, the outcome weighted learning framework \eqref{eq:owl} can also be applied to multiple treatments problems using angel based learning \cite{Zhang2017}.} On the other hand, researchers also propose another general framework in \cite{Zhang2012a,Zhang2012b}, which estimates the contrast function of the treatment effect directly. The proposed methods in \cite{Zhang2012a} and \cite{Zhang2012b} also enjoy the advantage of double robustness which allows for misspecification on either potential outcome model or propensity score model. Chen \cite{Chen2017} further extends the ideas in \cite{Zhao2012} and \cite{Zhang2012b} to more general loss functions.

However, different from traditional multi-class classification problem where each patient is assigned to one of the treatments, multi-label classification problem allows patients to be assigned to a combination of multiple treatments. A naive implementation called Label Powerset (LP) transforms a multi-label classification problem with $K$ treatments (or classes) into a multi-class classification problem with $2^K$ classes, therefore, the dimension increases exponentially with the number of treatment $K$. To avoid the curse of dimensionality and to increase efficiency, two strategies are often adopted. The first strategy is Binary Relevance (BR) \cite{Luaces2012}. It transforms the original problem to a $K$ independent binary classification problem. For example, suppose $\mathcal{A}=\{a_1,a_2,\cdots,a_K\}$, BR type of methods build $K$ classifiers to individually and independently decide whether $a_k$, $k=1,\cdots,K$, should be adopted. 
In this case, assumed linear classification rule with $p$-dimensional covariates, at least $Kp$ coefficients need to be estimated, which is much smaller than that of LP. However, it is broadly criticized by its independent assumption on treatments (or classes). 
Therefore, ranking by pairwise comparison \cite{Hüllermeier2008} is proposed. {\adapt However, this method only provides a ranking of the treatments. It does not provide any `zero' point to separate treatments which are beneficial from those are harmful.}

Besides those attempts to convert the multi-label classification problem into other problems,
an alternative is adopting an appropriate loss. One of the commonly adopted loss function is Hamming loss. 
Let $A=(A_1,\cdots,A_K)\in \mathcal{A}$ represents a vector of length $K$, where $\mathcal{A}=\{-1,1\}^K$. If the $k$th treatment is adopted, then $A_k=1$;  If not, then $A_k=-1$. Accordingly, decision rule $D(X)=(D_1(X),\cdots,D_K(X))\in \{-1,1\}^K$. Hamming loss of a decision rule $D(\cdot)$ given the sample $(X,A)$ is defined as
\begin{equation}
\frac{1}{K}\sum_{k=1}^{K}I\{A_k\not=D_k(X)\}.
\end{equation}
And it can be interpreted as the proportion of the misclassified labels. Say there are only 3 treatments (labels), the relationship between Hamming loss and 0-1 loss is shown in Figure~\ref{fig:compare_loss}. 

\begin{figure}[ht]
	\centering
	\includegraphics[scale=0.62]{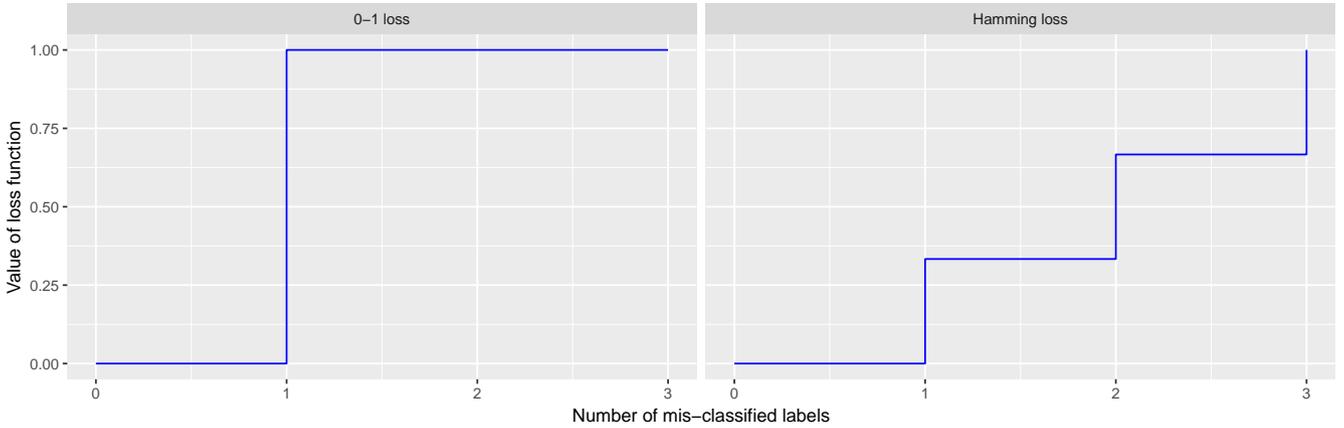}
	\caption{The left is the 0-1 loss, the right is the Hamming loss. Hamming loss is a step function and bounded by 0-1 loss.\label{fig:compare_loss}}
\end{figure}

Another interesting topic is the shared subspace. 
The idea proposed by \cite{Yan2007} assumes that all the decision rules are directly depend on the same subspace of covariates. Similar to central mean space estimation in sufficient dimension reduction \cite{Cook2005}, 
shared subspace essentially assumes that for certain $B$, there exist $\bar{D}_k$'s such that $D_k(X)=\bar{D}_k(B^\top X)$, for all $k\in 1,\cdots,K$. 
Through shared subspace, we can always borrow some efficiency from other treatments, and thus facilitate the estimation of unknown non-linear relationship.

In addition, scalability of the algorithm is also important because heavy computation complexity can undermine its practical application. 
Multi-label classification enables the scalable computation with respect to the number of treatments. The algorithm still needs to be scalable with the increasing sample size.

From this review of the literature, it's clear that 
the method we are looking for should fully address the following issues:
\begin{enumerate}
	\item \adapt{Applicability to precision medicine}.
	\item Appropriate loss with the framework of multi-label classification.
	\item Shared subspace.
	\item Scalable computation with both number of treatments and sample size.
\end{enumerate}

\section{Method}\label{sec3}

In this section, a loss function is proposed within outcome weighted learning framework, which takes advantage of multi-label classification. Our classifier based on neural networks (NN) is introduced and combined with the proposed loss function. NN as a classifier naturally satisfies the requirement of shared subspace. And its algorithm is also scalable with the sample size. 

\subsection{Outcome weighted learning with multi-label classification}

Let $X_i\in \mathcal{X}$ be the column vector of $p$-dimensional covariates for $i$th patient among total $n$ patients, $A_i=(A_{i1},\cdots,A_{iK})\in \{-1,1\}^K$ is the treatment assignment of the $i$th patient, and $R_i\in \mathbb{R}$ is the observed outcome of the $i$th patient. Again, decision rule is denoted as $$D(X)=(D_1(X),\cdots,D_K(X))\in \{-1,1\}^K.$$ 

Considering the following loss function for a given decision rule $D(X)$,
\begin{equation}\label{eq:proposed_loss}
L(D)=\frac{1}{n}\sum_{i=1}^{n}\frac{R_i}{\pi_{A_i}}\frac{1}{K}\sum_{k=1}^KI\{A_{ik}\not=D_k(X_i)\},
\end{equation}
where $\pi_{A_i}$ may be known in clinical trial or estimated in observational study. 
 When the problem is a multi-class classification problem, the loss above is exactly equal to the outcome weighted 0-1 loss proposed in \cite{Zhao2012}. When the problem is a multi-label classification problem, the proposed loss \eqref{eq:proposed_loss} is upper bounded by the outcome weighted 0-1 loss in outcome weighted learning. In addition, if $\frac{R_i}{\pi_{A_i}}$ is a constant, it reduces to Hamming loss. The loss function proposed in \eqref{eq:proposed_loss} combines the loss in outcome weighted learning and Hamming loss. Thus, we call it outcome weighted Hamming loss.

One of the difficulty in optimizing the proposed loss \eqref{eq:proposed_loss} is the non-smoothness and non-convexity of indicator functions. Hinge loss \cite{Cortes1995} or logistic loss can be adopted as surrogate losses for indicator functions. The fisher consistency of proposed loss in \eqref{eq:proposed_loss} and its surrogate loss are provided in the following sections. As proved in Section \ref{sec4}, our proposed loss is fisher consistent under small amount of interactions.

\adapt{
Another difficulty of the proposed method comes from the estimation of $\pi_{A_i}$. We provide two solutions to this issue when $K$ is large. First solution assumes that the treatment assignment $A$ is independent with the covariates $X$. In this case, $P(A=a|X)=P(A=a)$ can be estimated by the proportion of the patients with $A=a$ in the sample. Second solution assumes that each treatment assignment $A_k$ are independent with each other given covariates $X$. In this case, $P(A=a|X)=\prod_{k=1}^K P(A_k=a_k|X)$ and each $P(A_k=a_k|X)$'s can be estimated by logistic regression because $A_k$ are binary. When $K$ is small, multinomial regression can also be used to estimate $P(A=a|X)$ directly. For simplicity, the first solution is adopted in Section \ref{sec5} and Section \ref{sec6}.}

\subsection{Decision rule with deep learning}

In the previous section, outcome weighted Hamming loss is defined. Admittedly any suitable classifier can fit into our framework, the classifier adopted in this paper is the Neural Network (NN) for its advantages in shared subspace and scalable computation.

 To simply illustrate our idea, we start from a $3$-layer NN. The first layer on the bottom is the input layer where $X_i=(X_{i1},\cdots,X_{ip})^\top$ is the input vector of covariates. The layer in the middle is the hidden layer with $d$ hidden variables. The top layer consists of $K$ output variables. An example of the graph structure of this NN is presented in Figure~\ref{fig:NN}. In this toy example, only 3 treatments are considered, say treatment $E$, treatment $F$, and treatment $G$. Thus, $\mathcal{A}=\{-1,1\}^3$. For example, $(1,1,-1)$ represents $EF$ which is the combination of treatment $E$ represented by $(1,-1,-1)$, and $F$ represented by $(-1,1,-1)$. The $k$th output in the top layer is $\tilde{D}_k(X_i)$ given $X_i$ in the input layer. The sign of $\tilde{D}_k(X_i)$ indicates the treatment assignment of $k$th treatment, which is $D_k(X_i)$. As shown in Figure~\ref{fig:NN}, the adjacent layers are fully connected and no variables in the same layer are connected.
 
 Given the graphic structure described above, many NNs can fit into this framework. For example, Deep Belief Nets (DBN) proposed in \cite{Hinton2006}, which consists of stacks of Restricted Boltzmann Machines (RBMs) and a classifier based on the very top hidden layer. Another common choice is \adapt{Deep Neural Network (DNN)}, which is a large NN without probabilistic modeling. DNN can be obtained by firstly training a DBN and then fine-tuning by back-propagation. This approach of building DNN often times can help avoid local minimizers and obtain better generalization error. However, the topologies of DBN and DNN are slightly different. Although the skeleton for both are the same as shown in Fig~\ref{fig:DNN_DBN}, the directions of the connected lines can be different. In DBN, besides the top two layers, all connections between hidden layer and hidden layer or hidden layer and visible layer are bi-directional (or undirected), and only connections between top two layers are directed from the hidden layer to the output. In DNN, all connections are directed from lower layers to higher layers. Directed connections are defined by relationship similar to \eqref{eq:nn_build1} and \eqref{eq:nn_build2}, which have no probabilistic framework. Bi-directional (or undirected) connections are defined through undirected graphical models, such as RBM \cite{Goodfellow-et-al-2016}. Thus, the nodes connected by bi-directional (or undirected) connections are random variables. The value of the node is either a random sample from the defined conditional distribution or the conditional expectation given other nodes. For example, if the lower layer on the right in Fig~\ref{fig:DNN_DBN} is defined by RBM and denote the input layer as $v$, the first hidden layer as $h$, the joint density function of $(v,h)$ is defined as
 \begin{equation*}
 f_{(v,h)}(v,h)=\frac{1}{Z(\theta)}{\exp\{-E(v,h;\theta)\}},
 \end{equation*}
 where $\theta$ is unknown parameter, $Z(\theta)$ is a normalization constant, $E(v,h;\theta)$ is energy function which has certain forms depending on the type of the model \cite{Goodfellow-et-al-2016}. If $\theta$ is known, given an input $v$, the conditional distribution of $h$ given $v$ can be calculated. Either a random sample from this conditional distribution or conditional expectation of $h$ given $v$ can be used as the input for next layers. The density function of $v$ is 
 \begin{equation*}
 f_{v}(v)=\int\frac{1}{Z(\theta)}{\exp\{-E(v,h;\theta)\}}dh.
 \end{equation*}
 Then $\theta$ can be estimated by the maximize likelihood estimation (MLE) given observed data $v$. However, In this paper, we focus on DNN in both simulations and real data analysis.
 
 DNN is defined in the following fashion. For clarification, we explicitly define those variables and weights of connections given toy example in Fig \ref{fig:NN}. Let the input layer in the very bottom be $v=(v_1,\cdots,v_p)\in \mathbb{R}^p$, the hidden layer (consists of $d$ hidden variables) in the middle be $h=(h_1,\cdots,h_d)\in\mathbb{R}^d$, and the output layer on the very top be $\tilde{D}=(\tilde{D}_1,\cdots,\tilde{D}_K)\in\mathbb{R}^K$. Let $W_1$ be the $p\times d$ matrix representing the weights on connections between input layer and hidden layer. For example, the entry of $W_1$ on Row $3$ and Column $4$ is the weight on connection between $v_3$ and $h_4$. Similarly, let $W_2$ be the $d\times K$ matrix representing the weights on connections between hidden layer and output layer. For example, the entry of $W_2$ on Row $3$ and Column $2$ is the weight on connection between $h_3$ and $\tilde{D}_2$. Now, with these notations, the relationship between these variables can be defined as the following:
 \begin{eqnarray}\label{eq:nn_build1}
 h&=&ReLU(W_1^\top v+h_0),\\ \label{eq:nn_build2}
 \tilde{D}&=&W_2^\top h+\tilde{D}_0,
 \end{eqnarray}
 where $ReLU(\cdot)$ is ReLU function defined as $ReLU(t)=t1\{t>0\}$, for any $t\in \mathbb{R}$. $ReLU(O)$ with a vector $O$ represents a vector with ReLU function applied to each entry of $O$. $h_0$ is a constant $d$-dimensional vector and $\tilde{D}_0$ is a constant $K$-dimensional vector. Apparently, from the above relationship, $\tilde{D}$ can be written as $\tilde{D}(v)=(\tilde{D}_1(v),\cdots,\tilde{D}_K(v))$. Having the indicator function replaced by Hinge loss \cite{Cortes1995}, we can formulate the optimization problem as follows,
 \begin{equation}\label{eq:optimization}
 \mini_{\theta}\frac{1}{n}\sum_{i=1}^{n}\frac{R_i}{\pi_{A_i}}\frac{1}{K}\sum_{k=1}^K\left[1-A_{ik}\tilde{D}_k(X_i)\right]_+,
 \end{equation}
 where $\theta=(W_1,h_0, W_2,\tilde{D}_0)$ and $[\cdot]_+$ represents taking non-negative part. This optimization problem can be solved directly by back-propagation and SGD. Note that the structure of DNN and the activation function can be modified based on real data. The final estimated decision rule, $\hat{D}$, is the sign of $\tilde{D}$ given the minimizer of \eqref{eq:optimization}, $\hat{\theta}$.
 
 \begin{figure}
 	\centering
 	\includegraphics[scale=0.5]{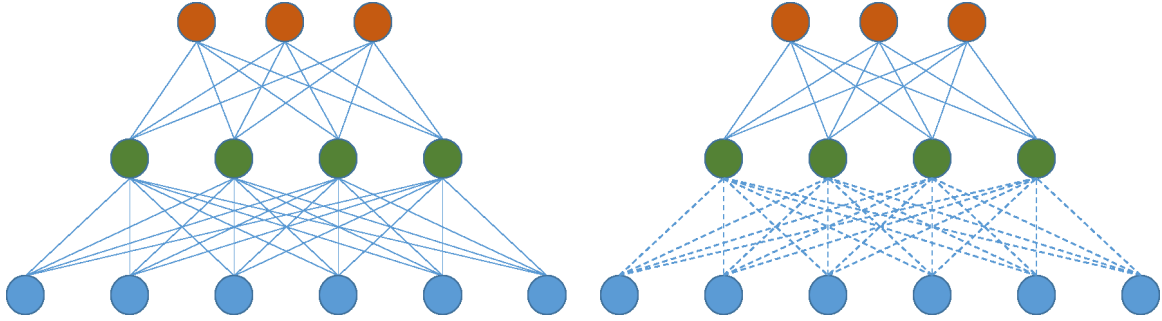}
 	\caption{The left is the structure of DNN and the right is that of DBN. The solid connections represent a directed edge from the lower layer to the top layer. The dashed connections represent bi-directional (or undirected) connections.\label{fig:DNN_DBN}}
 \end{figure}

NN is a natural choice in this setting for the following reasons. First, the optimization of NN given the proposed loss in \eqref{eq:proposed_loss} can be implemented by back-propagation \cite{Rumelhart1986} and stochastic gradient descent (SGD) \cite{Bottou2010}, which is scalable in terms of sample size due to the nature of SGD. Second, the decisions are made based on hidden layers which depend on some shared linear directions of $X_i$ and a pre-specified non-linear activation functions \cite{LeCun2015}. Thus, the subspace of these $K$ decision rules are shared. Third, the hidden variables can capture complicated correlations among treatments, which is quite clear in the point of graphical model that given $X_i$, $\tilde{D}_k$'s are not independent with each other. 
Based on universal approximation theorem proved in \cite{Barron1993}, hidden variables also introduce more flexibility into the model. Without the hidden layers, NN is equivalent to a linear model. 
In most applications, the structure of hidden layers such as the number of hidden layers and the number of hidden variables in each hidden layer is adjustable. More complicated the hidden structure is, more flexible the decision rule can be. Thus, our NN decision rule with proposed loss function in \eqref{eq:proposed_loss} can fully satisfy our requirements. Another ad hoc view of the hidden layers in Fig~\ref{fig:NN} is that those hidden variables may represent certain biological pathways through which certain treatment can affect the outcome. Thus, the decision rule is reasonable to depend on certain hidden variables.

\begin{figure}
	\centering
	\includegraphics[scale=0.5]{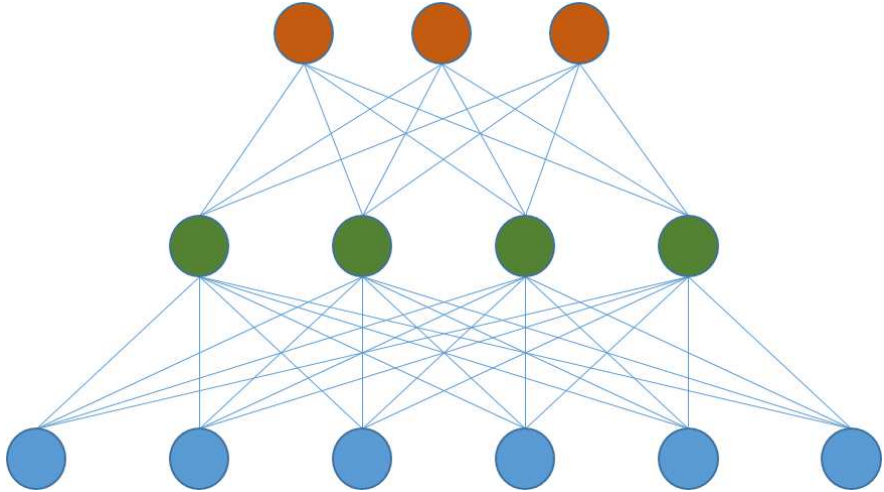}
	\caption{The first layer on the bottom is the input layer where $X_i=(X_{i1},\cdots,X_{ip})^\top$ is the input. The layer on the top is hidden layer with $d$ hidden variables. The top layer consists of $K$ output random variables.\label{fig:NN}}
\end{figure}

\subsection{Overfitting}

Due to strong representative power of DNN, how to avoid over-fitting has become a big issue. To avoid over-fitting, a commonly used strategy is regularization through penalization. Mathematically, we can consider to solve the following problem:
\begin{equation}
\mini_{\theta}\frac{1}{n}\sum_{i=1}^{n}\frac{R_i}{\pi_{A_i}}\frac{1}{K}\sum_{k=1}^K[1-A_{ik}\tilde{D}_k(X_i)]_++\lambda p(W_1,W_2),
\end{equation}
where $\theta=(W_1,h_0, W_2,\tilde{D}_0)$. $p(W_1,W_2)$ is the penalty on $W_1$ and $W_2$. For example, $p(W_1,W_2)=\|W_1\|_F^2+\|W_2\|_F^2$ is the ridge penalty, where $\|\cdot\|_F$ is the Frobenius norm, and $p(W_1,W_2)=\|W_1\|_1+\|W_2\|_1$ is lasso penalty, where $\|\cdot\|_1$ represents the sum of absolute value of all entries. Both ridge penalty and lasso penalty can shrink the weights towards $0$, but the advantage of lasso penalty is that it can produce sparse solutions. Another popular method to prevent over-fitting is the so-called dropout proposed in \cite{Srivastava2014}. Wager et al. \cite{Wager2013} argue that dropout is closely related to adaptive penalization and they also clarify its relationship with ridge penalty. 
In general, both regularization and dropout are possible choices to prevent over-fitting.

\adapt{In both regularization and dropout, how to choose the tuning parameter has been investigated for years\cite{Srivastava2014, tibshirani1996regression,zou2005regularization}. Cross-validation or training-validation-testing data split can be used to evaluate the prediction error and select tuning parameters with the best performance. }

\adapt{Other parameters such as number of layers and number of nodes each layer also play an important role in balancing under-fitting and over-fitting. Of course, cross-validation or training-validation-testing data split can be applied to tune the NN's structure, but it is extremely computational intensive and impractical some times. In general, how to design NN's structure is still an open problem, but extensive researches have been done in this field. On the one hand, some researchers are working on using genetic algorithm and reinforcement learning to decide these parameters \cite{leung2003tuning,zoph2016neural}. On the other hand, sparsity induced by lasso penalty can partially solve this issue. Suppose that all weights connected to a particular variable are $0$, it is equivalent to excluding this variable in the structure of NN. Additionally, other penalties, for example, in \cite{Scardapane:2017:GSR:3067301.3067328}, weights connected to a particular variable can be forced to be $0$ and thus the node is excluded in the structure. Thus, in practice, we may suggest a structure as more complex as possible considering the sample size and employ lasso penalty or other penalties to help select nodes included in the structure automatically.}

\section{Algorithm}


In this section, we introduce the algorithm to solve the proposed optimization problem. A straightforward solution is back-propagation with stochastic gradient descent (SGD). Back-propagation proposed in \cite{Goodfellow-et-al-2016} is an efficient algorithm to numerically compute the derivatives with respect to certain weight given a NN. SGD provides an computational effect alternative to gradient descent. Beyond directly implementing SGD, pre-training can also be used to prevent local minimizers and facilitate convergence of the algorithm. One of the pre-training procedure utilizes Restricted Boltzmann Machine (RBM) \cite{Goodfellow-et-al-2016}. In this procedure, we firstly train stacks of RBM in such fashion that the first RBM is trained on observed covariates and next RBM is trained on the top of the first one, taking the hidden layer in the first RBM as the visible layer in the next RBM, and so on. Then, for the two layers on the very top, a multi-label classifier based on outcome weighted Hamming loss is trained. After this pre-training procedure, the whole network is fine-tuned by back-propagation and SGD. 

\subsection{Stochastic gradient descent}

\adapt{
In this section, we briefly introduce the stochastic gradient descent (SGD) under general empirical risk minimization. Say we want to minimize the following loss function
\begin{equation}
\min_{\theta} \frac{1}{n}\sum_{i=1}^{n} l(Y_i, X_i;\theta),
\end{equation}
where $\theta$ is the general notation for all parameters of interest, $(X_i, Y_i) ,~i=1,\cdots,n$ is the observed data. In gradient descent, we have
\begin{equation}
\theta^{(m+1)}\leftarrow \theta^{(m)}-t\frac{1}{n}\sum_{i=1}^{n}\nabla_{\theta}l(Y_i, X_i;\theta),
\end{equation}
where $\theta^{(m)}$ is the value of $\theta$ at the $m$ step, and $t>0$ is a scalar called learning rate (step size). In stochastic gradient descent, we have
\begin{equation}
\theta^{(m+1)}\leftarrow \theta^{(m)}-t\frac{1}{|\mathcal{S}^{(m)}|}\sum_{i\in \mathcal{S}^{(m)}}\nabla_{\theta}l(Y_i,  X_i;\theta),
\end{equation}
where $\mathcal{S}^{(m)}$ is a random sample from $\{1,\cdots, n\}$ with or without replacement, and $|\mathcal{S}^{(m)}|$ is the cardinality of the set $\mathcal{S}^{(m)}$, which is a pre-specified batch size. It is easy to see that
\begin{equation}
E_{\mathcal{S}}\left [\frac{1}{|\mathcal{S}^{(m)}|}\sum_{i\in \mathcal{S}^{(m)}}\nabla_{\theta}l(Y_i,  X_i;\theta)\right ]=\frac{1}{n}\sum_{i=1}^{n}\nabla_{\theta}l(Y_i,  X_i;\theta),
\end{equation}
where $E_{\mathcal{S}}[\cdot]$ is the expectation of random sampling.} Because each update only depends on a small size of random sample, it is scalable with the increase of sample size.

\adapt{In each update of SGD, gradient is computed on a small batch based on a sub-sample of training dataset. On the one hand, the stochastic induced by sub-sampling prevents the algorithm from falling into local minimizers. On the other hand, it reduces the computation complexity dramatically compared with usual gradient descent. In general, SGD converges faster than usual gradient descent, especially for statistical problems \cite{bousquet2008tradeoffs}. However, due to the stochastic nature of SGD, the calculated gradient based on a batch is mostly not zero even at the global minimizer. Thus, the variance of sampling small batches has an impact on the updates of each iteration. To reduce this variance, SVRG and other techniques \cite{Johnson2013} are proposed. Another simple strategy is to gradually increase the size of the batch in gradient calculation and decrease learning rate (step size) exponentially in gradient update. The convergence of SGD has been proved in strictly convex problem \cite{bousquet2008tradeoffs}, while generally the convergence of SGD is still an open problem. }

\subsection{Implementation}

\adapt{In this section, we introduce some packages to implement SGD. A well known package in Python is called Tensorflow developed by Google. Tensorflow provides an extremely powerful tool to customize NN (number of layers and number of nodes each layer) and SGD (batch size and step size). It also includes some well-developed advanced algorithm for SGD and keeps updating. Compared with Tensorflow, Keras is a high-level NN API which is more user-friendly. It allows researchers to build a very flexible NN with a few lines of code. The down side is that it may not be easy to customize everything using Keras.  Recently, R initiated an access to Tensorflow, which is more friendly to statistical programmers. An experimental R package to implement the Keras is also available on GitHub. Readers can get  more information on R versions of Tensorflow and Keras on https://tensorflow.rstudio.com.} 



\section{Theoretical result}\label{sec4}

In this section, the theoretical justification of our proposed loss is provided. In the first part, we prove the fisher consistency of outcome weighted Hamming loss, which implies that minimizing the proposed loss is equivalent to minimizing the original outcome weighted 0-1 loss. In the second part, surrogate loss for outcome weighted Hamming loss is proposed and its consistency with outcome weighted Hamming loss is also provided, which indicates that minimizing the surrogate loss is also equivalent to minimizing the outcome weighted Hamming loss. Combining these two parts of theoretical results provides an overall theoretical support for our proposed method.

Without loss of generality, we assume that outcome $R$ is non-negative. 

\subsection{Fisher consistency of Hamming loss}

We establish the fisher consistency of outcome weighted Hamming loss in this section, indicating the equivalence of minimizing outcome weighted Hamming loss and outcome weighted 0-1 loss. Let $D^*(X)=(D_1^*(X), \cdots,D_K^*(X))\in \{-1,1\}^K$ be the decision rule that minimizes outcome weighted 0-1 loss:
\begin{equation}
\mathcal{R}(D)=E\left[\frac{R}{\pi_A}I\{A\not = D(X)\}\right].
\end{equation}
It is easy to see that $D^*(X)=\{a:\max_a E[R|A=a, X]\}$. Again, outcome weighted Hamming loss is defined as
\begin{equation}
\mathcal{R}_H(D)=E\left[\frac{R}{\pi_A}\frac{1}{K}\sum_{k=1}^KI\{A_k\not = D_k(X)\}\right].
\end{equation}

\begin{theorem}[Fisher consistency of outcome weighted Hamming loss]\label{thm1}
	Suppose $$D_k^*(X)=\sign\left\{\sum_{\{a:a_k=1\}}E[R|A=a,X]-\sum_{\{a:a_k=-1\}}E[R|A=a,X]\right\},$$
	then any function $f$ such that $$\mathcal{R}_H(f)=\inf_D\{\mathcal{R}_H(D)\}$$ satisfies 
	$$\mathcal{R}(f)=\inf_D\{\mathcal{R}(D)\}.$$
\end{theorem}
Note that in many cases, condition $D_k^*(X)=\sign\left\{\sum_{\{a:a_k=1\}}E[R|A=a,X]-\sum_{\{a:a_k=-1\}}E[R|A=a,X]\right\}$ holds. For example, a sufficient condition is when treatment effects are additive and there is no interaction, i.e. $E[R|A=a, X]=\sum_{\{k:a_k=1\}}T_{e_k}(X) + m(X)$, where $T_{e_k}, k=1,\cdots,K$, is the treatment effect given only treatment $k$, $e_k$ is a $K$-dimensional vector having the $k$th entry equals $1$ and all other entries equal $-1$. 
Moreover, small interactions are tolerable, as stated in Theorem \ref{thm2}. 

\begin{theorem}[Fisher consistency in special case]\label{thm2}
	Suppose that 
	\begin{equation*}
	E[R|A=a, X]=\sum_{\{k:a_k=1\}}T_{e_k}(X) + r_a(X) + m(X),
	\end{equation*}
	where $T_{e_k}(X)$ is the treatment effect of only $k$th treatment being adopted, $r_a(X)$ is the additional interaction, and $m(X)$ is the main effect. If $2\sup_a|r_a(X)|< \inf_k|T_{e_k}(X)|$ for all $X$, then any function $f$ such that $$\mathcal{R}_H(f)=\inf_D\{\mathcal{R}_H(D)\}$$ satisfies 
	$$\mathcal{R}(f)=\inf_D\{\mathcal{R}(D)\}.$$
\end{theorem}

Theorem \ref{thm2} provides a sufficient condition for fisher consistency of outcome weighted Hamming loss. It essentially provides a theoretical guarantee of the robustness of our method against small interactions. When the condition is violated, our method is not necessarily consistent and its performance depends on the magnitude of interactions and amounts of patient affected. In simulation, we compare the performance our method with other methods under small violation of this condition.
\subsection{Multi-label consistency of surrogate loss}

Minimizing the proposed loss \eqref{eq:proposed_loss} is still very hard due to the non-smoothness and non-convexity of indicator functions. Therefore, it is natural to replace the indicator functions in outcome weighted Hamming loss with some surrogate loss. For example, a common choice for outcome weighted Hamming loss is 
\begin{equation}\label{eq:surrogate_hamming}
\Phi_H(\tilde{D})=E\left[\frac{R}{\pi_A}\frac{1}{K}\sum_{k=1}^K\phi(A_k\tilde{D}_k(X))\right],
\end{equation}
where $\phi$ is a pre-defined convex function. Still, it is critical that minimizing the surrogate loss is equivalent to minimizing the original outcome weighted Hamming loss. Unfortunately, not every $\phi$ satisfies this condition. As shown in the following theorem, \eqref{eq:surrogate_hamming} is consistent with outcome weighted Hamming loss if $\phi$ is one of the following:
\begin{enumerate}
	\item Exponential: $\phi(x)=e^{-x}$;
	\item Hinge: $\phi(x)=(1-x)_+$;
	\item Least squares: $\phi(x)=(1-x)^2$;
	\item Logistic Regression: $\phi(x)=\ln(1+e^{-x})$.
\end{enumerate}

Formally, we have the following theorem.
\begin{theorem}[Multi-label consistency of surrogate loss]\label{thm3}
	Suppose $$\phi'(0)<0,$$
	for any function $\tilde{f}$ such that $$\Phi_H(\tilde{f})=\inf_{\tilde{D}}\{\Phi_H(\tilde{D})\},$$ let $f=\sign(\tilde{f})$, then $f$ satisfies 
	$$\mathcal{R}_H(f)=\inf_D\{\mathcal{R}_H(D)\}.$$
\end{theorem}

\section{Simulation}\label{sec5}

\subsection{Simulation with correctly specified model}\label{subsec5.1}
\label{Sec: correctmodel}
In this section, we will illustrate the performance of two DNN approaches under outcome weighted Hamming loss in \eqref{eq:optimization} by comparing with a naive method through simulations. In the following simulation settings, $K=5$, so in total there are $2^5=32$ combinations of treatments. The dimension of covariates is set to be $p=30$, which is common in clinical trials. As explained in Section \ref{review:multi-label}, the $K$ treatment multi-label problem can be decomposed to a multi-class problem with $2^K$ classes. The naive method further converts the multi-class problem to a series of two-class classification problem, where for each, it directly learns a binary classifier with linear decision rule through outcome weighted learning. To form a multi-class classifier with $32$ different classifiers with intercepts, $2^K(2^K-1)(p+1)/2=15376$ coefficients in $2^K(2^K-1)/2=496$ different linear decision rules have to be estimated. The first DNN approach, DNN-simple, is a DNN with only input and output layer. For any monotone activation function, DNN-simple is equivalent to $K=5$ binary linear classifiers with no shared subspace. In general, the only difference between DNN-simple and naive method is how to form the classifier for multi-label classification. In DNN-simple, only $K\times p+K=155$ parameters needs to be estimated. The comparison of DNN-simple and naive method essentially shows the efficiency boosting by adopting our proposed outcome weighted Hamming loss. The other DNN approach, DNN-1hdd, adds one hidden layer between input layer and output layer as shown on the left in Figure \ref{fig:DNN_DBN}. The hidden layer is fully connected with the input layer and the output layer.The NN structure allows for more flexibility and hence more parameters. In DNN-1hdd method, suppose the number of hidden variable is $n_h$, the total number of parameters to be estimated is  $p(n_h+1)+n_h(K+1)=p+n_h(p+K+1)=30+36n_h$. By comparing DNN-simple and DNN-1hdd, it is easy to tell the loss and gain to accommodate a more flexible model. In addition, the Bayes rule is also evaluated, simply to quantifies the signal-to-noise ratio in our simulation settings. The Bayes rule is the treatment assignment which gives the largest conditional expectation of the potential outcomes, i.e. 
\begin{equation}
D_{\rm B}(X)=\arg \max_aE\left[R|A=a,X\right].
\label{Eq:bayesrule}
\end{equation}
 The Bayes rule is impossible to implement in reality, but in simulations, since we know the data generating procedure, we can directly evaluate $E\left[R|A=a,X\right]$. 

In the simulation, every patient has the same probability to receive one of the treatment combination, and for each treatment, the treatment effect is generated from a NN with one hidden layer as shown in the left of Figure~\ref{fig:DNN_DBN}. The number of the hidden variables in the hidden layer is $n_h=45$. Firstly, we define the treatment effect of the $k$th treatment for $i$th patient $T_{k,i}$. Let 
\begin{eqnarray}
h&=&ReLU(W_1^\top X_i),\\
T_{k,i}&=&\sign\left (W_{2,\cdot,k}^\top h\right )\left\{\frac{0.05\exp\left\{W_{2,\cdot,k}^\top h\right\}}{1+\exp\left\{W_{2,\cdot,k}^\top h\right\}}+2.0\right\},
\end{eqnarray}
where $h$ is the hidden layer which is a $n_h$-dimensional vector, $W_{2,\cdot,k}$ is the $k$th column of $W_2$, and $\sign(\cdot)$ is the function of taking sign coordinate-wise. All entries in $W_1$ and $W_2$ are generated from a standard normal distribution independently. Secondly, the main effect is defined by the following:
\begin{equation}
M_i=0.05\frac{\exp\left\{\gamma^\top X_i\right\}}{1+\exp\left\{\gamma^\top X_i\right\}}-2.05,
\end{equation}
where $\gamma$'s are the coefficients for main effect, whose entries are also generated from a standard normal distribution independently. Let $A_i\in\{-1,1\}^K$ represents the combination of assigned treatment to $i$th patient. Note that $R_{i,A_i}$, the potential outcome when given $A_i$ is defined by the following:
\begin{equation}
R_{i,A_i}=\sum_{\{k: A_{ik}=1\}}T_{k,i}+M_i+\sigma\epsilon_{i,A_i},
\end{equation}
where $\epsilon_{i,A_i}$ follows standard normal distribution and is independent with any other random variables. It can be simply verified that $\forall \beta_{p,k}, \beta_{n,k}, $ and $\gamma$, given any $A_i$, $\sum_{\{k: A_{ik}=1\}}T_{k,i}+M_i\in[-12,8]$. Let $A_i^{\rm opt}$ represents the selected combination of treatment, given which the potential outcome is maximized. In our setting, $\sigma$ is chosen to be $1.1$. In fact, it can be observed that the Bayes rule as defined in Equation \ref{Eq:bayesrule} is the sign of the $T_{k,i}$ for each $k$. In addition, based on Theorem \ref{thm2} and \ref{thm3}, our proposed loss is consistent with outcome weighted 0-1 loss in this setting.

In both DNN approaches, the activation function between hidden layer and output layer, or input layer and output layer is fixed to be a simple centered monotone transformation. The activation function between input layer and hidden layer in DNN-1hdd is chosen to be ReLU. $L_1$ penalty is applied to all the weights in both of DNN approaches.

To compare different methods, three scores are defined to quantify their performance. Two of the scores are based on misclassification rate. Let $\hat{A}_i$ be the predicted combination of treatments. Usually, the misclassification rate is defined as
\begin{equation}
MCR=\frac{1}{n}\sum_{i=1}^n 1\{\hat{A}_i\not = A_i^{\rm opt}\}.
\end{equation}
To account for the fact that $P(A_i^{\rm opt}=a)$ is not the same, for all $a\subset\{1,\cdots,K\}$, the average of proportion of misclassification rate is also proposed as following
\begin{equation}
AMCR=\frac{1}{2^K}\sum_{a\in\{-1,1\}^K}\frac{\sum_{i=1}^n 1\{\hat{A}_i\not=a, A_i^{\rm opt}=a\}}{\sum_{i=1}^n 1\{A_i^{\rm opt}=a\}}.
\end{equation}
To adjust MCR and AMCR based on the total number of combinations of treatments and make their performance comparable with binary classifier for binary classification problem, adjusted MCR and AMCR is calculated by $1-(1-MCR)^{1/K}$ and $1-(1-AMCR)^{1/K}$, respectively. Because in the framework of personalized medicine, the ultimate goal is not classification, but better clinical outcome. Thus, it is primary to consider the average benefit which defines as following
\begin{equation}
\label{Eq:AB}
AB=\frac{1}{n}\sum_{i=1}^n R_{i,\hat{A}_i}.
\end{equation}
For MCR (or adjusted MCR) and AMCR (or adjusted AMCR), lower is better. However, higher AB is preferred.

The simulation procedure works as follows. First, a training set and a validation set with the same sample size $n_{train}K$ are generated separately, and a sequence of candidate tuning parameters are pre-specified, in this case, 0.1, 0.01, 0.001, 0.0001. Then for each candidate tuning parameter, the Naive method, DNN-simple, and DNN-1hdd are trained on the training dataset and tested on the validation set, where we compute the misclassification rate on the validation dataset for each method. The tuning parameter gives the lowest misclassification rate on the validation set is selected, respectively for each method. At last, we evaluate the three scores (MCR, AMCR, AB) under the selected tuning parameter on testing dataset which is independently generated with sample size $n_{test}=10n_{train}K$. The above procedure is repeated 100 times. The mean and standard error (SE) of the scores are reported in Table \ref{tab2}. Figure \ref{Fig:comparison} shows a summary for adjusted scores and AB over 100 repeats. 

Overall, DNN-1hdd performs the best in terms of MCR (or adjusted MCR) and AB. Note that Bayes method is infeasible in practice, its superiority is because it takes advantage of the true model. The Bayes method serves as a reference to quantify the best possible performance we could get in this problem. Although DNN-simple can only produce linear decision rules, which don't include the true decision rule, its performance is still competitive with DNN-1hdd because linear decision rule can explain the true non-linear model to some extent, especially in small samples when the information is limited. The fact that DNN-1hdd allows more flexibility also introduces additional variation. In general, DNN-1hdd performs slightly better than DNN-simple with this bias and variance trade-off. In terms of AMCR (or adjusted AMCR), DNN-simple performs the best among all three methods. From Table \ref{tab2} and Figure \ref{Fig:comparison}, as the sample size increases, the misclassification decreases and average benefit increases for all methods, which agrees with our consistency result proved in Theorem \ref{thm2}. However, the decreasing in MCR (or adjusted MCR) becomes slower as sample size increases, because (1) the scores are lower bounded by Bayes rule; (2) the algorithm is terminated after a certain times of iterations and takes the final update as the minimizer. In practice, one possible strategy is to use a validation dataset to monitor the algorithm and pick the update with the best performance in validation dataset. 

\begin{center}
	\begin{table}[t]%
		\centering
		\caption{Comparison of naive method, DNN-simple, and DNN-1hdd by adjusted scores.\label{tab2}}%
		\begin{tabular*}{500pt}{@{\extracolsep\fill}ccccc@{\extracolsep\fill}}
			\toprule
			$n_{train}$ & Method & adjusted MCR & adjusted AMCR & AB\\
			\midrule
			200 & Naive &0.4708(0.0026) &0.4652(0.0012) &-1.5171(0.1004)\\
			200 & DNN-simple & 0.2824(0.0022)  &0.3555(0.0022)& 0.5765(0.1172)\\
			200 & DNN-1hdd & 0.2782(0.0028)  &0.3709(0.0025)& 0.6188(0.1200)\\ 
			200 & Bayes & 0.0657(0.0001)&0.1108(0.0027) &3.1173(0.1143)\\ \hline
			1000 & Naive &0.3517(0.0024) &0.3729(0.0012) &-0.0897(0.0948)\\
			1000 & DNN-simple & 0.2125(0.0016)  &0.3008(0.0028)& 1.3609(0.1153)\\
			1000 & DNN-1hdd & 0.2142(0.0015)  &0.3240(0.0041)& 1.3376(0.1145)\\ 
			1000 & Bayes & 0.0657(0.0001)&0.1111(0.0027) &3.0871(0.1133)\\ \hline
			4000 & Naive &0.3083(0.0029) &0.3396(0.0018) &0.3502(0.0883)\\
			4000 & DNN-simple & 0.1907(0.0014)  & 0.2902(0.0034) & 1.7257(0.1166)\\
			4000 & DNN-1hdd & 0.1806(0.0012) & 0.2991(0.0050) & 1.8435(0.1134)\\ 
			4000 & Bayes & 0.0656(0.0000)&0.1140(0.0029) &3.1965(0.1127)\\  \hline
			10000 & Naive &0.2996(0.0035) &0.3305(0.0015) &0.4934(0.0766)\\
			10000 & DNN-simple & 0.1846(0.0014)  & 0.2865(0.0033) &1.7272(0.1008)\\
			10000 & DNN-1hdd & 0.1617(0.0012) &0.2889(0.0063)  &1.9858(0.0986)\\ 
			10000 & Bayes & 0.0657(0.0001)&0.1111(0.0027) & 3.1272(0.0977)\\
			\bottomrule
		\end{tabular*}
	\end{table}
\end{center}
\begin{figure}
	\centering
	\includegraphics[height=9cm, width=12cm]{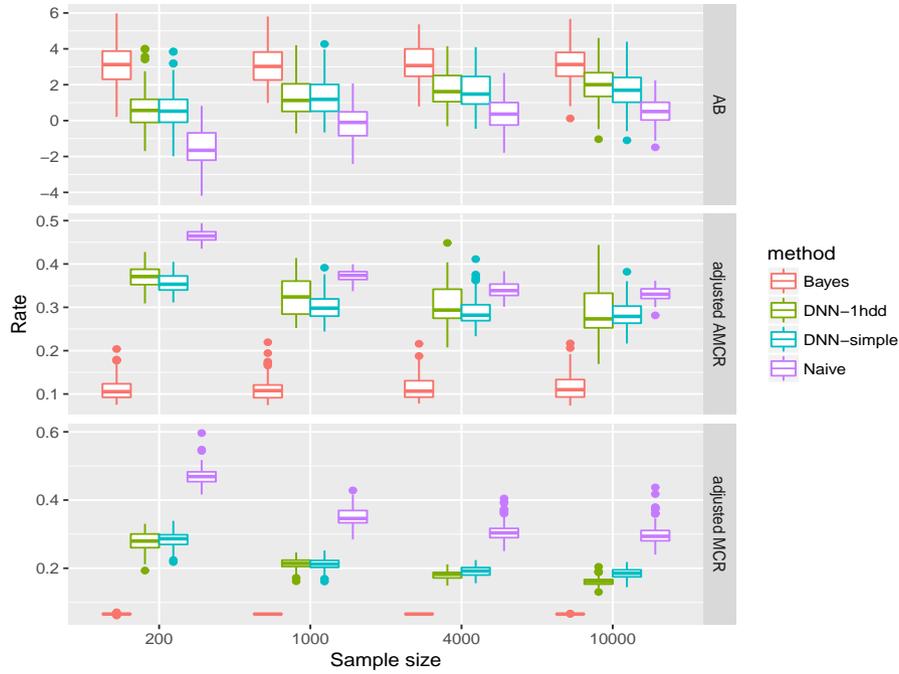}
	\caption{Comparison of naive method, DNN-simple, and DNN-1hdd by adjusted MCR, AMCR, and AB.\label{Fig:comparison}}
\end{figure}

\begin{figure}[ht]
	\centering
	\includegraphics[height=9cm, width=12cm]{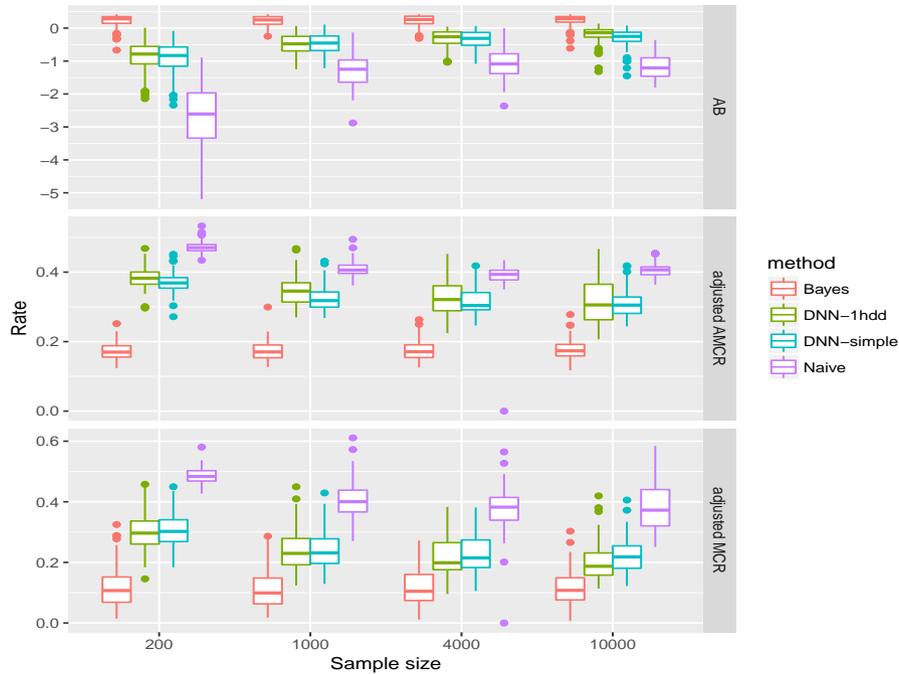}
	\caption{Comparison of naive method, DNN-simple, and DNN-1hdd by adjusted MCR, AMCR, and AB with model mis-specification.\label{Fig:comparison_mis}}
\end{figure}

\subsection{Simulation with model mis-specification}

In this section, we compare our two DNN approaches with a naive method when the condition in Theorem \ref{thm2} is not satisfied. The simulation procedure is the same as in Section \ref{Sec: correctmodel} except for the way to generate $R_{i, A_i}$. In this simulation, given $T_{k,i}$ and $M_i$, the potential outcome is generated by
\begin{equation}
R_{i,A_i}=\sum_{\{k: A_{ik}=1\}}T_{k,i}-\gamma (\sum_{\{k: A_{ik}=1\}}T_{k,i})^2+M_i+\sigma\epsilon_{i,A_i},
\end{equation}
where $\gamma=0.1$ and $\sigma=0.2$. In this setting, it is easy to see that $\sup_A|\gamma (\sum_{\{k: A_{ik}=1\}}T_{k,i})^2|\approx10$ and $\inf_k|T_{k,i}|\approx2$. Thus, the condition in Theorem \ref{thm2} is violated. The simulation procedure is the same as that described in Section \ref{subsec5.1}.

The averages of Adjust MCR, AMCR, and AB over 100 repeats with their standard deviations are reported in Table \ref{tab2_mis} and Figure \ref{Fig:comparison_mis}, which show that our proposed methods DNN-1hdd and DNN-simple outperform Naive method with respect to all three scores. Overall, DNN-1hdd has the best performance with respective to AB and adjusted MCR, while DNN-simple performs the best when considering the adjusted AMCR. Although the Naive method is consistent, it performs very bad under finite sample because of the high dimensionality of this problem. Our two DNN approaches still perform well when the small interaction assumption is violated.

\begin{center}
	\begin{table}[t]%
		\centering
		\caption{Comparison of naive method, DNN-simple, and DNN-1hdd by adjusted scores with model mis-specification.\label{tab2_mis}}%
		\begin{tabular*}{500pt}{@{\extracolsep\fill}ccccc@{\extracolsep\fill}}
			\toprule
			$n_{train}$ & Method & adjusted MCR & adjusted AMCR & AB\\
			\midrule
			200 & Naive &0.4863(0.0024) &0.4720(0.0017) &-2.6427(0.0921)\\
			200 & DNN-simple & 0.3048(0.0053)  &0.3703(0.0029)& -0.9197(0.0516)\\
			200 & DNN-1hdd & 0.2994(0.0060)  &0.3829(0.0028)& -0.8746(0.0491)\\ 
			200 & Bayes & 0.1157(0.0064)&0.1726(0.0027) &0.2124(0.0199)\\ \hline
			1000 & Naive &0.4074(0.0063) &0.4096(0.0022) &-1.2951(0.0470)\\
			1000 & DNN-simple & 0.2432(0.0063)  &0.3244(0.0036)& -0.4939(0.0308)\\
			1000 & DNN-1hdd & 0.2433(0.0067)  &0.3462(0.0043)& -0.4885(0.0307)\\ 
			1000 & Bayes & 0.1127(0.0065)&0.1732(0.0029) &0.2136(0.0164)\\ \hline
			4000 & Naive &0.3761(0.0070) &0.3898(0.0044) &-1.1061(0.0424)\\
			4000 & DNN-simple & 0.2323(0.0064)  & 0.3166(0.0040) & -0.3592(0.0267)\\
			4000 & DNN-1hdd & 0.2254(0.0070) & 0.3275(0.0051) & -0.3151(0.0247)\\ 
			4000 & Bayes & 0.1199(0.0066)&0.1761(0.0028) &0.2293(0.0161)\\  \hline
			10000 & Naive &0.3890(0.0084) & 0.4059(0.0020)& -1.1827(0.0370)\\
			10000 & DNN-simple &0.2210(0.0053) &0.3102(0.0037)&-0.3050(0.0279)\\
			10000 & DNN-1hdd &0.2016(0.0060)&0.3128(0.0060)&-0.1966(0.0247)\\ 
			10000 & Bayes &0.1146(0.0054) &0.1761(0.0027)& 0.2354(0.0172)\\ 
			\bottomrule
		\end{tabular*}
	\end{table}
\end{center}

\section{Real data analysis}\label{sec6}

In this section, we apply our method to an electronic health record (EHR) data for type 2 diabetes patients from Clinical Practice Research Datalink (CPRD). 1139 patients are included in the dataset. For each patient, 21 covariates are collected before treatment assignment, which include demographical variables such as gender, BMI, HDL, and LDL, and also indicator of complications such as stroke and hypertension. The primary endpoint is change in A1c. Because A1c typically drops after applying the treatments, the primary endpoint is always negative and smaller is preferable. Thus, the proposed method can be easily implemented by considering negative change in A1c as patient outcome.

In this dataset, 4 treatments are considered for each patient, DDP4, sulfonylurea (SU), metformin (MET), and TZD. These 4 treatments function via four different biological processes. DDP4 increases incretin levels, which inhibits glucagon release. SU increases insulin release from $\beta$-cell in pancreas. MET decreases glucose production by the liver and increase the insulin sensitivity of body tissue. TZD makes cells more dependent on oxidation of carbohydrates. These 4 treatments target on different cells or functional organs, so it is reasonable to assume little treatment interaction among them and additive treatment effects.

In this real dataset, $K=4$, resulting in $16$ possible combinations of treatments. Because the original average benefit score can not be directly calculated in real settings, the naive method, DNN-1hdd, and DNN-simple are evaluated under a weighted version of average benefit defined as follows, 
\begin{eqnarray}
T&=&\frac{\sum_{i=1}^{n} w_iR_i}{\sum_{i=1}^{n} w_i},\nonumber
\end{eqnarray}
where $w_i=\frac{I\{\hat{A}_i=a_i\}}{P(A_i=a_i|X_i)}$, $a_i$ is the observed treatment assignment, $R_i$ is the negative change in A1c, $\hat{A}_i$ is the predicted treatment assignment for the $i$th patient. Intuitively, $T$ is the weighted average of change in A1c over all the patients with the same treatment assignment as our predicted treatment assignment, which estimates the change in A1c if the fitted treatment assignment is adopted. If our proposed methods can recover the underlying optimal decision rule to some extent, their $T$ are expected to be lower than that of the naive method, indicating higher efficacy of our treatment recommendations. In addition, the number of patients $N$ satisfying $a_i=\hat{A}_i$ is also reported.

Multiple imputation is adopted to deal with missing data. Because of the randomness of the multiple imputation, $5$ different imputed datasets are analyzed following the same procedure, and then the scores from these $5$ imputed data are summarized by average. 

For each imputed dataset, we do the following. Firstly, the whole dataset is randomly split into two datasets, training set and testing set. Training set contains $912$ patients, and testing dataset contains $227$ patients. All three methods are fitted on training set respectively, and the score for each of the method is calculated based on the testing set. This procedure is repeated $100$ times, each time the score is recorded. After summarizing these scores across $5$ imputed datasets, the mean of $T$, $N$ and their standard errors (SEs) are reported in Table~\ref{tab3}, and all results over $100$ repeats are shown by boxplots in Figure~\ref{Fig:realdata}.

It can be observed that both DNN-simple and DNN-1hdd have lower $T$ compared with naive method. Moreover, both DNN methods show significant effect on A1c while naive method does not. Comparing between two DNN methods, DNN-1hdd performs slightly better than DNN-simple in terms of $T$. In addition, $N$, which is the size of the subgroup with the same treatment assignment as predicted for our proposed methods, is much larger than that for naive method. Overall, our proposed methods have better performance than naive method, and DNN-1hdd is slightly better than DNN-simple.

\begin{center}
\begin{table}[t]
	\caption{Comparison of Naive method, DNN-simple, and DNN-1hdd by scores in real data example. $T$ is the weighted average of change in A1c over all the patients with the same treatment assignment as the method suggests. $N$ is the number of the patients whose treatment assignment coincide with the predicted.\label{tab3}}%
	\begin{tabular*}{500pt}{@{\extracolsep\fill}ccc@{\extracolsep\fill}}
		\toprule
		Method & $T$ & $N$ \\ \midrule
		Naive & -1.534(0.078) & 4.410(0.107)\\
		DNN-simple & -2.605(0.058) & 23.058(0.371) \\
		DNN-1hdd & -2.695(0.057) &  25.790(0.394)\\
		\bottomrule
	\end{tabular*}
\end{table}
\end{center}

\begin{figure}
	\centering
	\includegraphics[height=9cm, width=12cm]{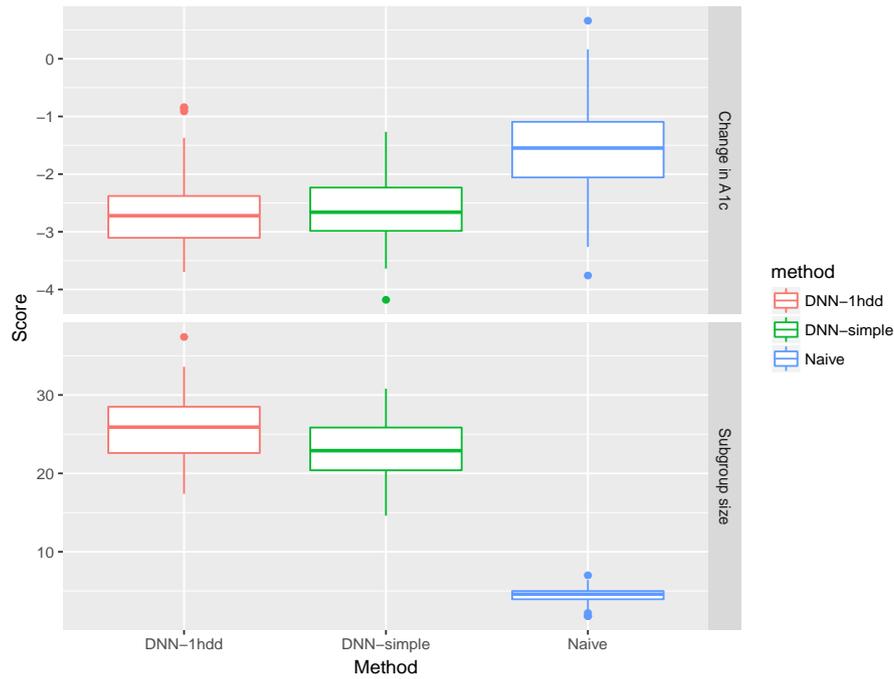}
	\caption{Boxplots of scores for Naive method, DNN-simple, and DNN-1hdd over 100 repeats .\label{Fig:realdata}}
\end{figure}

\section{Conclusions and discussions}\label{sec7}

In this paper, an outcome weighted deep learning framework is proposed to estimate optimal combination therapies. Both simulation and real data analysis provide solid evidence on the power of our proposed method. Essentially, the proposed loss, outcome weighted Hamming loss, can be applied to any occasion, even when deep learning is not a desired classifier. For example, linear classifiers can be used in the case when efficiency and convexity is very important. Although other nonlinear classifiers can also be adopted, when adopting these methods, advantages of deep learning such as sharing subspace may disappear unless inducing certain techniques such as dimension reduction. Deep learning framework can also be applied under other losses. For example, partial ranking loss proposed in \cite{Gao2013} can also be combined with deep learning approach. However, for ranking based loss function, one of the significant drawback is the lack of `zero' point to distinguish between good treatments and bad treatments. In other word, ranking based loss function can only provide a rank of treatments instead of recommending treatment. In this paper, the advantages of adopting deep learning approach and the loss function has been articulated . Our method enjoys all of these advantages.

Furthermore, the proposed method is critical for future research. It can be extended in multiple directions. The first possible extension focuses on loss functions. In previous sections, it is easy to observe that Hamming loss provides an approximation to the original 0-1 loss. As we have shown, this approximation is done by overlooking strong interactions among treatments. While original 0-1 loss is very flexible so that it is generally consistent, the huge number of parameters to be estimated may undermine the efficiency and lead to huge computation costs. Thus, a natural extension of our method is to design a family of loss functions such that our method can be adaptive to certain amount of interactions among treatments. For example, Hamming loss is the proportion of mis-classified labels, and can be rewritten into the following
\begin{equation*}
	1-\frac{1}{K}\sum_{k=1}^K 1\left\{A_k=D_k\right\}.
\end{equation*}
0-1 loss can also be rewritten into the following
\begin{equation*}
1-\prod_{k=1}^K 1\left\{A_k=D_k\right\}.
\end{equation*}
Naturally, a family of loss functions can be defined as
\begin{equation*}
1-\frac{1}{\binom{\tau}{K}}\sum_{\{k_1\cdots,k_{\tau}\}\subset\{1,\cdots,K\}} \prod_{q=1}^{\tau}1\left\{A_{k_q}=D_{k_q}\right\},
\end{equation*}
where $\tau$ is a parameter. We call it $\tau$th order Hamming loss. When $\tau=1$, it is the same as Hamming loss, which has the most strict conditions on treatment interactions in order to guarantee its fisher consistency. When $\tau=K$, it is actually 0-1 loss, which has the fewest conditions on treatment interactions in order to guarantee its fisher consistency. With the increasing of $\tau$, the loss function can accommodate more and more interactions, but may lead to more and more parameters to be estimated. Thus, the existence of $\tau$ allows us to choose loss function adaptively to the data. 

The second possible extension focuses on the dynamic assignment of multiple treatments. Combination therapies considered in this paper do not involve treatment transition problem. For the real data analysis and simulation, all the treatment are assumed to be applied to the patient at the same time. However, this is not always true in real life. Patients typically change from one treatment to another. Thus, instead of deciding a static treatment assignment, a dynamic treatment assignment with transition scheduling is a more realistic and reasonable solution to precision medicine. During this process, multiple outcomes may also be involved in the analysis. The process of single outcome over time may also play an important role in this extension. 

The third extension focuses on how to combine our proposed technique with methods in \cite{Zhang2012a} and \cite{Zhang2012b} such that we can directly estimate the contrasts of treatment effects. In this case, it is critical to model the potential outcomes and propensity scores in an efficient way, especially for combination therapies. In general, our work in this paper is the keystone to all these potentials.

One of the limitations in our proposed method is that our method may fail under large interactions between treatments. When the interactions between treatments share the same direction of treatment recommendations, our method still holds, even the condition in \ref{thm2} fails. For example, in the simulation with model mis-specification, if $\gamma\leq 0$, it can be shown that our method is still consistent but the condition in \ref{thm2} fails. When the interactions between treatments are large and have the opposite direction of treatment recommendations, for example, it is extremely harmful to take two good medications simultaneously. Each combination therapy should be considered as a totally new treatment whose effect is irrelevant to the treatment effects when taking medications separately. In this case, treatment recommendation should be considered as a multi-class classification problem rather than a multi-label classification problem, and no information can be borrowed to improve efficiency.


\section*{Acknowledgments}

This work was one of research topics in Eli Lilly and Company summer intern program. All supports were provided by Eli Lilly and Company. Thank Yebin Tao for suggestions on this work.

\appendix

\section{Proof of Theorems\label{app1}}

\begin{proof}[Proof of Theorem~\ref{thm1}]
	Let us compute the minimizer for the proposed Hamming loss.
	\begin{eqnarray}
	\mathcal{R}_H(D)&= &E\left [\frac{R}{\pi_A}\frac{1}{K}\sum_{k=1}^{K}I\{A_k\not = D_k(X)\}\right] \nonumber\\
	&= &E\left\{\frac{1}{K}\sum_{k=1}^K \left[\sum_{\{a:a_k=1\}}E[R|A=a,X]I\{1\not = D_k(X)\}+\sum_{\{a:a_k=-1\}}E[R|A=a,X]I\{-1\not = D_k(X)\}\right]\right\}.\nonumber
	\end{eqnarray}
	It is easy to see that the minimizer of the proposed Hamming loss, $f=(f_1,\cdots,f_K)\in \{-1,1\}^K$, can be written as
	$f_k=\sign\left\{\sum_{\{a:a_k=1\}}E[R|A=a,X]-\sum_{\{a:a_k=-1\}}E[R|A=a,X]\right\}, \forall k\in\{1,\cdots,K\}$. Thus, the conclusion follows given the condition that $D_k^*(X)=\sign\left\{\sum_{\{a:a_k=1\}}E[R|A=a,X]-\sum_{\{a:a_k=-1\}}E[R|A=a,X]\right\}$.
\end{proof}

\begin{proof}[Proof of Theorem~\ref{thm2}]
	We claim that the $D^*(X)$ in this case is the combination of all the treatments with $T_{e_k}(X)>0$. For any other decision rule $D(X)\not = D^*(X)$,
	\begin{align*}
	&E[R|A=D^*(X), X]-E[R|A=D(X), X]\\
	=&\sum_{\{k:T_{e_k}(X)>0\}}T_{e_k}(X) + r_{D^*}(X) - \sum_{\{k:D_k(X)=1\}}T_{e_k}(X) - r_{D}(X)\\
	\geq & \inf_k |T_{e_k}(X)|-2\sup_a|r_a(X)|\\
	>&0.
	\end{align*}
	Thus, our claim is true, $D^*_k(X)=\sign(T_{e_k})$. Now, let us compute the following.
	\begin{align*}
	&\sum_{\{a:a_k=1\}}E[R|A=a,X]-\sum_{\{a:a_k=-1\}}E[R|A=a,X]\\
	=&\sum_{\{a:a_k=1\}}\left\{\sum_{\{k':a_{k'}=1\}}T_{e_{k'}}(X) + r_a(X)\right\} - \sum_{\{a:a_k=-1\}}\left\{\sum_{\{k':a_{k'}=1\}}T_{e_k'}(X) + r_a(X)\right\}\\
	\geq&\sum_{\{a:a_k=1\}}\left\{\sum_{\{k':a_{k'}=1, k\not=k'\}}T_{e_{k'}}(X)\right\}-\sum_{\{a:a_k=-1\}}\left\{\sum_{\{k':a_{k'}=1, k\not=k'\}}T_{e_{k'}}(X)\right\}\\
	&+2^{K-1}T_{e_k}(X)-2^K\sup_a|r_a(X)|\\
	=&2^{K-1}[T_{e_k}(X)-2\sup_a|r_a(X)|]
	\end{align*}
	Thus, $\sign\left\{\sum_{\{a:a_k=1\}}E[R|A=a,X]-\sum_{\{a:a_k=-1\}}E[R|A=a,X]\right\}=\sign(T_{e_k})$. The conclusion follows.
\end{proof}

\begin{proof}[Proof of Theorem~\ref{thm3}]
	Let us compute
	\begin{eqnarray}
	\Phi_H(\tilde{D})&= &E\left[\frac{R}{\pi_A}\frac{1}{K}\sum_{k=1}^K\phi(A_k\tilde{D}_k(X))\right] \nonumber\\
	&= &E\left\{\frac{1}{K}\sum_{k=1}^K \left [\sum_{\{a:a_k=1\}}E[R|A=a,X]\phi(\tilde{D}_k(X)) +\sum_{\{a:a_k=-1\}}E[R|A=a,X]\phi(-\tilde{D}_k(X))\right ]\right\}. \nonumber
	\end{eqnarray}
	Let $\eta = \frac{\sum_{\{a:a_k=1\}}E[R|A=a,X]}{\sum_{\{a:a_k=1\}}E[R|A=a,X] + \sum_{\{a:a_k=-1\}}E[R|A=a,X]}$, by \cite{Bartlett2006}, the conclusion follows.
\end{proof}

\nocite{*}
\bibliographystyle{wileyNJD-AMA}
\bibliography{ref}%

\clearpage
\end{document}